\documentclass[english]{IEEEtran}
\usepackage{multirow}
\usepackage{booktabs}
\usepackage{lipsum}
\usepackage{amsfonts}
\usepackage{placeins}
\usepackage{array}
\usepackage{amsmath,cases,amssymb,nccmath}
\usepackage{amsmath}
\usepackage{graphicx}
\usepackage{cite}
\usepackage{subfigure}
\usepackage{epsfig}
\usepackage{url}
\usepackage{algorithm}
\usepackage{algorithmic}
\usepackage{epstopdf}
\usepackage{balance}
\usepackage{color}
\usepackage{setspace}
\usepackage{ mathrsfs }
\usepackage{algorithm}
\DeclareGraphicsExtensions{.eps,.pdf}

\newtheorem{lemma}{Lemma}
\newtheorem{remark}{Remark}

\newcommand{\st}{{\mathrm{s.t.}}}

\newcommand{\bh}{\mathbf{h}}
\newcommand{\bg}{\mathbf{g}}
\newcommand{\bI}{\mathbf{I}}
\newcommand{\bp}{\mathbf{p}}

\newcommand{\bw}{\mathbf{w}}

\newcommand{\Sk}{\mathtt{S}_k}
\newcommand{\Dk}{\mathtt{D}_k}

\newcommand{\R}{\mathtt{R}}

\newcommand{\ID}{\mathtt{ID}}
\newcommand{\EH}{\mathtt{EH}}

\newcommand{\K}{\mathcal{K}}

\newcommand{\bwi}{\mathbf{w}^{(\kappa)}}

\newcommand{\betai}{\beta^{(\kappa)}}
\newcommand{\pki}{p_k^{(\kappa)}}

\newcommand{\varthetai}{\vartheta^{(\kappa)}}

\newcommand{\alphadi}{\alpha_2^{(\kappa)}}

\makeatletter
\g@addto@macro\normalsize{%
	\setlength\abovedisplayskip{1.5pt}
	\setlength\belowdisplayskip{2.5pt}
	\setlength\abovedisplayshortskip{1.5pt}
	\setlength\belowdisplayshortskip{2.5pt}
}
\makeatother

\newcommand*{\hili}{\color{black}}

\renewcommand{\IEEEbibitemsep}{0pt plus 0.5pt}
\makeatletter
\IEEEtriggercmd{\reset@font\normalfont\fontsize{7.5pt}{8pt}\selectfont}
\makeatother
\IEEEtriggeratref{1}

\begin{document}

\title{ Optimization of Rate Fairness in Multi-Pair Wireless-Powered Relaying Systems}
\author{Van-Phuc Bui, Van-Dinh Nguyen, Hieu V. Nguyen, Octavia A. Dobre, and Oh-Soon Shin\\
 \thanks{This research was supported in part by Basic Science Research Program through the National Research Foundation of Korea (NRF) funded by the Ministry of Education (No. 2017R1D1A1B03030436),  in part by  the Luxembourg National Research Fund (FNR) in the framework of the FNR-FNRS bilateral project ``InWIP-NET: Integrated Wireless Information and Power Networks,'' and in part by the Natural Sciences and Engineering Research Council of Canada (NSERC), through its Discovery program. (\emph{Corresponding author: Oh-Soon Shin}.)
 	
 V.-P. Bui, H. V. Nguyen, and O.-S. Shin are with the Department of ICMC Convergence Technology and School of Electronic Engineering, Soongsil University, Seoul 06978, South Korea (e-mail: vanphucbui@soongsil.ac.kr; hieuvnguyen@ssu.ac.kr; osshin@ssu.ac.kr).
 
V.-D. Nguyen was with the Department of ICMC Convergence Technology, Soongsil University, Seoul, South Korea. He is now with the Interdisciplinary Centre for Security, Reliability and Trust (SnT) – University of Luxembourg, L-1855 Luxembourg (email: dinh.nguyen@uni.lu).
 
 O. A. Dobre is with the Faculty of Engineering and Applied Science, Memorial University, St. John’s, NL A1X3C5, Canada (e-mail: odobre@mun.ca).
 \vspace{-10pt}
}
\vspace{-10pt}
}
\maketitle
\vspace{-10pt}
\begin{abstract}
This letter considers a multi-pair decode-and-forward relay network where a power-splitting (PS) protocol is
adopted at the energy-constrained relay to provide simultaneous
wireless information and energy harvesting (EH). To achieve
higher efficiency of EH, we propose a new PS-based EH architecture at the relay by incorporating an alternating current (AC)
computing logic, which is employed to directly use the wirelessly
harvested AC energy for computational blocks. Under a nonlinear
EH circuit, our goal is to maximize the fairness of end-to-end rate
among user pairs subject to power constraints, resulting in a non-convex problem. We propose an iterative algorithm to achieve
a suboptimal and efficient solution to this challenging problem
by leveraging the inner approximation framework. Numerical
results demonstrate that the proposed algorithm outperforms the
traditional direct current computing and other baseline schemes.
\end{abstract}
\vspace{-10pt}

\vspace{-0.3cm}
\section{Introduction}\label{sec:intro}
Wireless relays have been considered to improve the spectral
efficiency and reliability, and to extend the coverage area of
wireless networks. Among numerous proposed relaying
protocols, decode-and-forward (DF) and amplify-and-forward
are most widely studied in the literature. The former has
drawn considerable attention due to its superior performance
compared to the latter \cite{Rankov:JSAC:07}.

With the dramatic growth of user devices, especially those
with low-cost and low-power requirements, we can envisage
future networks employing wireless relays capable of using
harvested power for information forwarding, rather than depending
on the grid power supply. To this end, simultaneous
wireless information and power transfer (SWIPT) technique is
an effective means to realize both energy harvesting (EH) and
information decoding from the transmitted radio-frequency
(RF) signals, prolonging the network lifetime of relays \cite{NasirTWC13,ClerckxJSAC19,COMML:Octavia}.
Various SWIPT-based relay schemes based on time-switching
relaying (TSR) and power-splitting relaying (PSR) have been
proposed, including multiple-input multiple-output (MIMO)
DF relay \cite{BenkhelifaJSAC16}, self interference-aided EH relaying \cite{ZhangTWC17}, and relays
with interference alignment \cite{ChuTCOM18}. The common approach in
the aforementioned works is that the EH power circuit converts
the harvested alternating current (AC) power to direct current
(DC) power to assist user data transmission and activate basic
functions (i.e., operating circuits and computational blocks).
Given that the wireless EH performance is very limited due to
high path-loss in far-field transmission, the use of DC computing
(DCC) results in significant system performance loss. The
reason is that the conversion efficiency of current rectifiers is
relatively low (i.e., about $50\sim60\%$). Fortunately, the works in
\cite{WanESL17} and \cite{SalmanVLSI18} have demonstrated through practical experiments that the AC power harvested from the RF signals can be
directly used to activate computational blocks. The benefit of
using AC computing (ACC) was first revealed in downlink SWIPT
 \cite{TranCOMML19} and NOMA-SWIPT networks \cite{NguyenACCESS2019}.

Motivated by the above discussion, we study a multi-pair
wireless-powered relaying system, where a multiple-antenna
DF relay receives both information and energy from source
nodes in the first phase and then utilizes the energy to forward
the information to destination nodes in the second phase. Contrary
to the previous works on SWIPT-based relay networks
\cite{NasirTWC13,BenkhelifaJSAC16,ZhangTWC17,ChuTCOM18}, this letter poses the following completely new
issues: $(i)$ A novel PSR architecture-enabled ACC is proposed
by leveraging charge-recycling theory, which aims at using the
EH more efficiently due to its low-power consumption and
no conversion loss; $(ii)$ Successive interference cancellation
(SIC) technique is adopted at the information decoding (ID) receiver \cite{TsePramodBook05}, which
is capable of improving both spectral efficiency and user
fairness. {\hili We consider a new problem
of max-min end-to-end (e2e) rate among user pairs under
a practical model of EH circuit \cite{XiongTWC17}, which is formulated
as a non-convex program.} Towards an efficient solution, we
convert the original problem into an equivalent non-convex
problem in a more tractable form, and then develop a lowcomplexity
iterative algorithm with convergence guaranteed.
By leveraging the inner approximation (IA) framework \cite{Beck:JGO:10},
the proposed algorithm solves a second order cone program
(SOCP) at each iteration, which is very efficient for practical
implementations. Numerical results are provided to confirm
that our proposed algorithm is efficient in terms of the e2e
rate fairness.

\vspace{-5pt}
\section{System Model and Problem Formulation}\label{sec:sys_model}

We consider a multi-pair wireless-powered relaying network, which consists of one energy-constrained DF relay $\R$ equipped with $N$ antennas and the set  $\mathcal{K}\triangleq\{1,\dots,K\}$ of $K=|\mathcal{K}|$ single-antenna user pairs, as illustrated in Fig. \ref{fig:systemodel}. In the $k$-th user pair, we assume that the source node $\Sk$ communicates with the destination node $\Dk$  via $\R$ and  there is no direct link between $\Sk$ and $\Dk$ due to path-loss and shadowing. The channels from $\Sk\rightarrow\R$ and $\R\rightarrow\Dk$ are denoted by $\bh_k\in\mathbb{C}^{N\times1}$ and $\bg_k\in\mathbb{C}^{1\times N}$, respectively, which 
are assumed to change block-by-block. The transmission block time, denoted by $T$, is divided into two phases: SWIPT phase $(\tau T)$ from $\{\Sk\}_{k\in\mathcal{K}}\rightarrow\R$ and wireless information transfer (WIT) phase $(1-\tau)T$ from $\R\rightarrow\{\Dk\}_{k\in\mathcal{K}}$, where $\tau\in(0,1)$ is a fraction of block time.

\vspace{-5pt}
\subsection{System Model}
\begin{figure}[t]
	\center
	\includegraphics[width=0.28\textwidth,trim={0.0cm 0.0cm 0cm -0.0cm}]{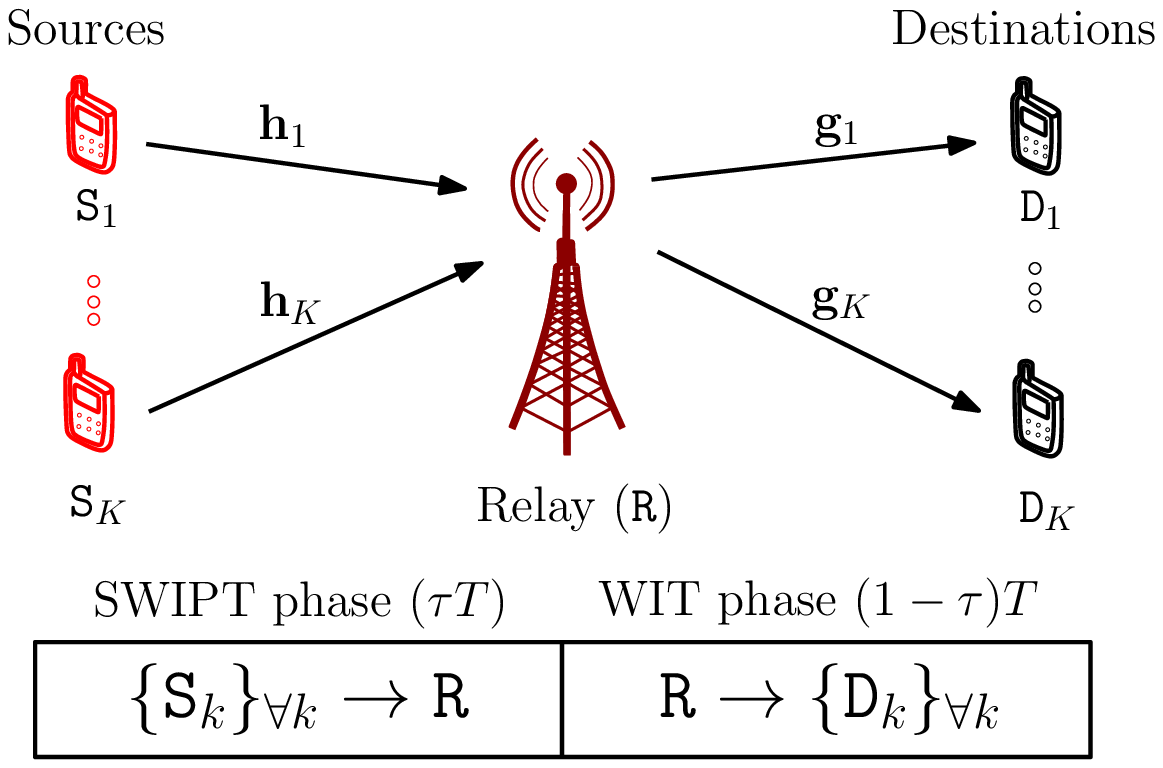}
	\caption{A multi-pair DF relaying network with SWIPT  and WIT  phases.}
	\label{fig:systemodel}
	\vspace{-10pt}
\end{figure}
\begin{figure}[t]
	\centering
	\includegraphics[width=0.28\textwidth,trim={0.0cm 0.0cm -0cm -0.0cm}]{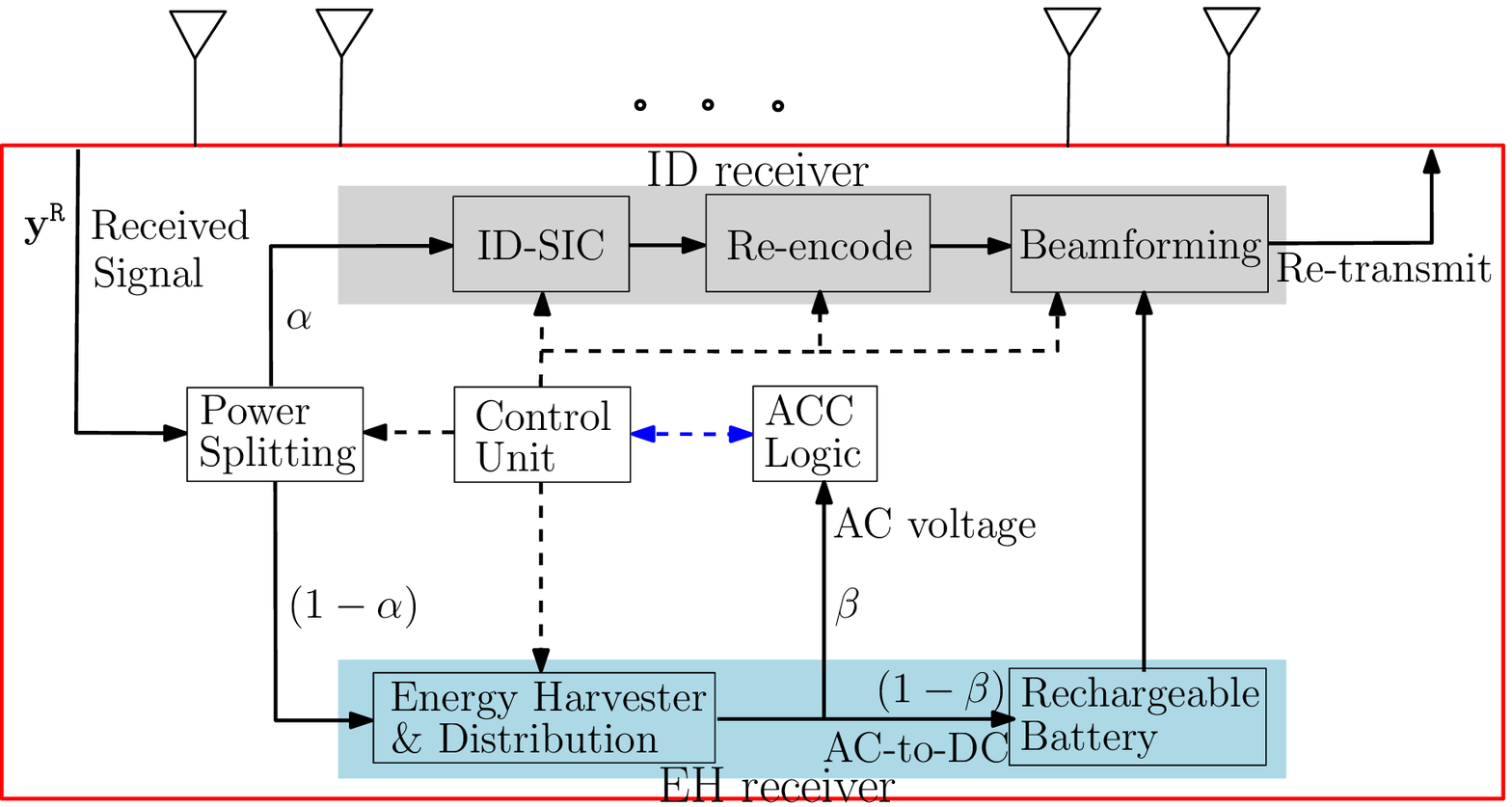}
	\caption{Proposed PSR architecture-enabled SIC and ACC at  relay.}
	\label{fig:relayoperation}
\end{figure}

\subsubsection{SWIPT Phase}
{\hili In Fig. \ref{fig:relayoperation}, we propose a new PSR architecture which enables the SIC technique and ACC logic at the ID and EH receivers, respectively.  In particular,  the received RF signal
	at $\R$ is split into two parts: ID and EH signals. In the EH receiver, the energy harvesting and  distribution blocks
	split the harvested AC power into two
	flows: one to directly supply the wirelessly harvested AC
	power for the ACC logic without rectification and regulation, while other to charge the battery for transmitting
	signals in the WIT phase by using the AC-to-DC rectifier. Note that the use of ACC logic eliminates the EH conversion loss.} Let $\alpha\in(0,1)$ be a portion of the RF  signal $\mathbf{y}_\mathtt{R}$ received at the relay using  a power splitter. The ID and EH signals  can be expressed as:
\begin{IEEEeqnarray}{rCl}
	\mathbf{y}_{\mathtt{R}}^{\ID} = \sqrt{\alpha}\mathbf{y}_\mathtt{R} + \mathbf{n}_{\mathtt{R}}\quad \text{and}\quad \mathbf{y}_{\mathtt{R}}^{\EH} = \sqrt{1-\alpha}\mathbf{y}_\mathtt{R},
	\label{eq:yRL}
\end{IEEEeqnarray}
where $\mathbf{y}_\mathtt{R} = \sum_{k\in \K} p_k\bh_k s_k + \mathbf{n}_{\mathtt{Ant}}$. Here $p_k$ and $s_k$ with $\mathbb{E}\{|s_k|^2\} = 1$ are the transmit power coefficient and the transmitted symbol at $\Sk$, respectively;  $\mathbf{n}_{\mathtt{Ant}}\sim \mathcal{CN}(0,\sigma^2_{\mathtt{Ant}}\bI)$ and $\mathbf{n}_{\mathtt{R}}\sim \mathcal{CN}(0,\sigma^2_{\mathtt{R}}\bI)$  are  the antenna noise and additional circuit noise introduced by the ID receiver, which are modeled as additive white Gaussian noise (AWGN) \cite{ShiTWC2014}. Without loss of generality, we normalize $T$ to 1 and rearrange the users in the ascending order of their channel gains, i.e., $\|\bh_1\|^2_2\leq\cdots\leq\|\bh_k\|_2^2\cdots\leq\|\bh_K\|^2_2$. We adopt the minimum mean square error and SIC (MMSE-SIC) technique at the ID receiver to decode  signals from sources \cite{TsePramodBook05}. To enhance the user fairness, we assume that the decoding order of SIC is from $s_K$ to $s_1$. Hence, the data rate (measured in  nats/sec/Hz) in decoding $s_k$ at $\mathtt{R}$ is given as
\begin{IEEEeqnarray}{rCl}
R_{1,k}(\bp,\tau,\alpha) = \tau\ln\bigl(1 + 	\gamma_{1,k}(\bp,\alpha)\bigl),
	\label{eq:rateSkR}
\end{IEEEeqnarray}
where ${\gamma}_{1,k}( \mathbf{p},\alpha) =  p_k^2\bh_k^H\boldsymbol{\Phi}_k^{-1}\bh_k$
with $\boldsymbol{\Phi}_k \triangleq\sum\nolimits_{\ell=1}^{k-1}p_\ell^2\bh_\ell\bh_\ell^H  + \sigma_{\mathtt{Ant}}^2\mathbf{I} + \frac{\sigma^2_{\mathtt{R}}}{\alpha}\mathbf{I}$  and $\bp\triangleq\{p_k\}_{k\in\K}$.

Next, the EH signal is further split into two flows by the energy harvester and distribution block, which are $\sqrt{\beta}\mathbf{y}_{\mathtt{R}}^{\EH}$ with the fraction $\beta\in(0,1)$ to directly supply an AC voltage to the ACC and the remaining $\sqrt{1-\beta}\mathbf{y}_{\mathtt{R}}^{\EH}$ to be rectified to the DC power. The DC power is stored in a rechargeable battery to be used for the data transmission in the WIT phase. By the charge-recycling theory \cite{WanESL17,SalmanVLSI18}, the average harvested AC power supplying ACC  can be expressed as
\begin{IEEEeqnarray}{rCl}
	P_\R^\mathtt{ACC}(\bp,\tau, 1-\alpha,\beta) = \tau(1-\alpha)\beta\sum_{k\in \K} p_k^2\|\bh_k\|^2_2.
	\label{eq:PACC}
\end{IEEEeqnarray}
Considering a realistic nonlinear EH model \cite{XiongTWC17}, the average harvested DC power at the EH receiver can be calculated as
{\small\begin{IEEEeqnarray}{rCl}
	&&P_\R^\mathtt{DC}(\bp,\tau,1-\alpha,1-\beta) = \tau\mfrac{\bar{P}_\mathtt{EH}^{\max}}{1-\Omega} \nonumber\\
	&&\qquad\qquad\quad\times\Big(\mfrac{1}{1+\exp\bigl(-a(P_\mathtt{R}^{\mathtt{IN}}(\bp,1-\alpha,1-\beta)-b)\bigr)}-\Omega\Big),
	\label{eq:PDC}
\end{IEEEeqnarray}}
\hspace{-7pt} where $P_\mathtt{R}^{\mathtt{IN}}(\bp,1-\alpha,1-\beta) \triangleq (1-\alpha)(1-\beta)\sum\nolimits_{k\in \K} p_k^2\|\bh_k\|_2^2$ is the AC power at the input of the EH circuit, $\bar{P}_\mathtt{EH}^{\max}$ is the maximum harvested power,  the constants $a$ and $b$ specify the EH circuits, and $\Omega=\bigr(1+\exp(ab)\bigl)^{-1}$.

\subsubsection{WIT Phase} The relay re-encodes  signals and forwards them to the destinations using the harvested power in the SWIPT phase. The re-encoded signal of the $k$-th pair is linearly weighted with the beamformer  $\bw_k\in\mathbb{C}^{N\times 1}$ at the relay prior to being forwarded to $\Dk$.
 As a result, the data rate decoded by  $\Dk$ is given as
\begin{IEEEeqnarray}{rCl}
	R_{2,k}(\bw,1-\tau) = (1-\tau)\ln\bigl(1 + 	\gamma_{2,k}(\bw)\bigr),
	\label{eq:rateDk}
\end{IEEEeqnarray}
where $\gamma_{2,k}(\bw) \triangleq \frac{|\bg_k\bw_k|^2}{\sum\nolimits_{\ell\in \K \backslash k}|\bg_{k}\bw_{\ell}|^2 + \sigma_k^2}$, $\bw \triangleq \{\bw_k\}_{k\in\mathcal{K}}$ and $\sigma_k^2$ is the variance of the AWGN at $\Dk$.
\begin{remark}
	The use of power domain-based NOMA at destinations \cite{NguyenACCESS2019}
	is inefficient since the harvested energy at the relay is very limited. It is noted that in multi-user SWIPT-based relay
	networks, the better the EH performance is, the more severer the
	network interference at the relay is.
\end{remark}

\subsection{Optimization Problem Formulation}
The achievable e2e rate of the $k$-th  pair can be  defined as
\begin{IEEEeqnarray}{rCl}
	R_k = \min\big\{R_{1,k}(\bp,\tau,\alpha),R_{2,k}(\bw, 1-\tau)\big\},\quad \forall k\in\mathcal{K}.
	\label{eq:rate}
\end{IEEEeqnarray}
The average power consumed by the relay  can be expressed as
\begin{IEEEeqnarray}{rCl}
	{ P_\R^{\mathtt{tot}}(\bw,1-\tau) = (1-\tau)P_\R^\mathtt{BF}(\bw)   + P_\R^\mathtt{sta}, }
	\label{eq:Ptot1}
\end{IEEEeqnarray}
where $P_\R^\mathtt{BF}(\bw) \triangleq\sum\nolimits_{k\in \K}\|\bw_k\|^2_2$ and $P_\R^\mathtt{sta}$ are the radiated power in the WIT phase and  the static power consumed by the circuits at $\R$, respectively.

We can observe that all parameters are mutually dependent, and thus,  should be jointly optimized.
The optimization problem of maximizing the minimum e2e rate among all  user pairs can be mathematically expressed as
\begingroup
\allowdisplaybreaks\begin{subequations}\label{MMR}
	\begin{IEEEeqnarray}{lcll} 
		& \max \limits_{\bp, \bw,\tau,\alpha,\beta}&\quad & r_0\triangleq\min\limits_{k\in\mathcal{K}}\ \{R_k\}  \label{MMRa}\\
		&\st&& P_\R^{\mathtt{tot}}(\bw,1-\tau)  \leq P_\R^\mathtt{DC}(\bp,\tau,1-\alpha,1-\beta), \label{MMRb}\qquad \\
		&&&  P_\R^\mathtt{ACC}(\bp,\tau, 1-\alpha,\beta)  \geq P^\mathtt{ACC}_{\min}, \label{MMRc}\\
		&&&p_k^2 \leq P_{\Sk}^{\max},\quad \forall k\in\K \label{MMRd}\\
		&&&  \tau\in(0,1),\ \alpha\in(0,1),\ \beta\in(0,1).\label{MMRe}
	\end{IEEEeqnarray}
\end{subequations}\endgroup
Here constraint \eqref{MMRb} ensures that the power consumption cannot exceed the harvested power.  $P^\mathtt{ACC}_{\min}$ in \eqref{MMRc} is the minimum  AC power  required for the ACC, while  \eqref{MMRd} represents the transmit power constraint at the source nodes.
\begin{remark}
The optimization problem with the traditional DCC can also be formulated as
	{\small\begin{IEEEeqnarray}{lcll} \label{MMRDCC}
		&& \max \limits_{\bp, \bw,\tau,\alpha,\beta=0}\quad  \tilde{r}_0\triangleq\min\limits_{k\in\mathcal{K}}\ \{R_k\},\quad\st\ \eqref{MMRd}, \eqref{MMRe},  \nonumber\\
		&&\qquad \&\ P_\R^{\mathtt{tot}}(\bw,1-\tau) + P^\mathtt{DCC}_{\min} \leq P_\R^\mathtt{DC}(\bp,\tau,1-\alpha,1), \label{MMRDCCc}\qquad 
	\end{IEEEeqnarray}}
\end{remark}
where $\beta=0$ and $P^\mathtt{DCC}_{\min}$ is the minimum DC power required for DCC. In addition to no EH conversion loss by the use of ACC, the benefit of problem \eqref{MMR} over \eqref{MMRDCC} can be realized by the fact that  $P^\mathtt{DCC}_{\min} $ is much higher than $P^\mathtt{ACC}_{\min}$ \cite{SalmanVLSI18}.

\vspace{-5pt}
\section{Proposed Algorithm}\label{section_alg}

\subsection{Equivalent Formulation}
Problem \eqref{MMR} is a non-convex program due to the  non-concave and non-smooth objective \eqref{MMRa} and non-convex constraints \eqref{MMRb} and \eqref{MMRc}. {\hili A direct application of the proposed method in \cite{NguyenACCESS2019}  for solving \eqref{MMR} still involves a nonconvex problem due to strong coupling between optimization
variables, and thus, several preliminary steps are necessary.} For that, we first make change of variables as $\tau_1 = \tau^{-1}, \tau_2 = (1-\tau)^{-1}, \alpha_1 = \alpha^{-1},  \alpha_2 = (1-\alpha)^{-1}$, and introduce the slack variables $\psi_{1,k}> 0,{\psi}_{2,k}> 0,\forall k$ and $r\geq 0$ to  rewrite \eqref{MMR} into the following equivalent
problem:
	\begin{IEEEeqnarray}{rCl} \label{MMReqi}
		 &&\max_{\bp, \bw,\boldsymbol{\tau}, \boldsymbol{\alpha},\boldsymbol{\psi},\beta,r}\quad  r  \IEEEyessubnumber\label{MMReqi:a}\\
		&&\quad\st\ \tau_i^{-1}\ln\bigl(1+ \psi_{i,k}^{-1}\bigr)  \geq r,\ \forall i\in\mathcal{I}\triangleq\{1,2\}, k\in\mathcal{K},      \IEEEyessubnumber\label{MMReqi:b}\qquad\\
		&&\qquad\quad \gamma_{1,k}(\bp,\alpha_1^{-1}) \geq  \psi_{1,k}^{-1},\ \forall k\in\mathcal{K},      \IEEEyessubnumber\label{MMReqi:c}\\
		&&\qquad\quad \gamma_{2,k}(\bw)  \geq  \psi_{2,k}^{-1},\ \forall k\in\mathcal{K},      \IEEEyessubnumber\label{MMReqi:d}\\
		&&\qquad\quad P_\R^{\mathtt{tot}}(\bw,\tau_2^{-1})  \leq P_\R^\mathtt{DC}(\bp,\tau_1^{-1},\alpha_2^{-1},1-\beta), \IEEEyessubnumber\label{MMReqi:e}\qquad \\
		&&\qquad\quad  P_\R^\mathtt{ACC}(\bp,\tau_1^{-1}, \alpha_2^{-1},\beta)  \geq P^\mathtt{ACC}_{\min}, \IEEEyessubnumber\label{MMReqi:f}\\
		&&\qquad\quad \tau_1^{-1} + \tau_2^{-1} \leq 1, \tau_1 >1 ,\tau_2 >1,\IEEEyessubnumber\label{MMReqi:g}\\
	  &&\qquad\quad\alpha_1^{-1} + \alpha_2^{-1} \leq 1, \alpha_1 >1, \alpha_2 >1,\IEEEyessubnumber\label{MMReqi:h}\\
		&&\qquad\quad \beta\in(0,1),\ p_k^2 \leq P_{\Sk}^{\max},\quad \forall k\in\K, \IEEEyessubnumber\label{MMReqi:i}
	\end{IEEEeqnarray}
where $\boldsymbol{\tau}\triangleq\{\tau_i\}_{i\in\mathcal{I}}$, $\boldsymbol{\alpha}\triangleq\{\alpha_i\}_{i\in\mathcal{I}}$ and $\boldsymbol{\psi}\triangleq\{\psi_{i,k}\}_{i\in\mathcal{I},k\in\mathcal{K}}$.
We now provide the following lemma to characterize the key property of problem \eqref{MMReqi}.
\begin{lemma}\label{lemma:1}
The optimization problems \eqref{MMR} and \eqref{MMReqi} are equivalent, as they share the same optimal
solution set and  objective value (i.e., $r_0^\star = r^\star$).
\end{lemma}
\begin{proof} See Appendix.\end{proof}

\subsection{IA-based Iterative Algorithm}
In problem \eqref{MMReqi},  the non-convex constraints include \eqref{MMReqi:b}-\eqref{MMReqi:f}. Let us handle the non-convex constraint \eqref{MMReqi:b} first. We can see that the function $f(\psi_{i,k},\tau_i)\triangleq\tau_i^{-1}\ln\bigl(1+ \psi_{i,k}^{-1}\bigr)$ is convex on the domain $(\tau_i >1, \psi_{i,k} > 0)$, which is useful to develop an approximate solution by the IA method. At iteration $\kappa$ of an iterative algorithm presented shortly, \eqref{MMReqi:b} is innerly approximated   as
\begin{IEEEeqnarray}{rCl}\label{MMReqi:bconvex}
f^{(\kappa)}(\psi_{i,k},\tau_i) \triangleq A^{(\kappa)} + B^{(\kappa)}\psi_{i,k} + C^{(\kappa)}\tau_i \geq r,\ \forall i,k,\quad
\end{IEEEeqnarray} 
where $A^{(\kappa)} \triangleq 2\ln\big(1+1/\psi_{i,k}^{(\kappa)}\big)/\tau_i^{(\kappa)} + 1/(\psi_{i,k}^{(\kappa)}+1)\tau_i^{(\kappa)}$, $B^{(\kappa)} \triangleq -1/\psi_{i,k}^{(\kappa)}(\psi_{i,k}^{(\kappa)}+1)\tau_i^{(\kappa)}$ and $C^{(\kappa)}\triangleq -\ln\big(1+1/\psi_{i,k}^{(\kappa)}\big)/(\tau_i^{(\kappa)})^2$  are constant. Note that $f^{(\kappa)}(\psi_{i,k},\tau_i)$ in \eqref{MMReqi:bconvex} is concave and represents a global lower bound of $f(\psi_{i,k},\tau_i)$ at the feasible point $(\psi_{i,k}^{(\kappa)},\tau_i^{(\kappa)})$ \cite[Appendix A]{NguyenJSAC18}, satisfying $f^{(\kappa)}(\psi_{i,k}^{(\kappa)},\tau_i^{(\kappa)})=f(\psi_{i,k}^{(\kappa)},\tau_i^{(\kappa)})$.

Next, we tackle the non-convexity of \eqref{MMReqi:c} and \eqref{MMReqi:d}. For \eqref{MMReqi:c}, we first consider the function $h(p_k,\bar{\boldsymbol{\Phi}}_k)\triangleq p_k^2\mathbf{h}_k^H\bar{\boldsymbol{\Phi}}_k^{-1}\mathbf{h}_k$ with $p_k>0$ and $\boldsymbol{\bar{\Phi}}_k\triangleq\sum\nolimits_{\ell=1}^{k-1}p_\ell^2\bh_\ell\bh_\ell^H  + \sigma_{\mathtt{Ant}}^2\mathbf{I} + \alpha_1\sigma^2_{\mathtt{R}}\mathbf{I}\succ\mathbf{0}$. A concave approximate function of $h(p_k,\bar{\boldsymbol{\Phi}}_k)$ is given as
\begin{IEEEeqnarray}{rCl}\label{eq:approhfunctio}
h(p_k,\bar{\boldsymbol{\Phi}}_k) \geq h^{(\kappa)}(p_k,\bar{\boldsymbol{\Phi}}_k)\triangleq 2p_k^{(\kappa)}\mathbf{h}_k^H(\bar{\boldsymbol{\Phi}}_k^{(\kappa)})^{-1}\mathbf{h}_kp_k \nonumber\\
- (p_k^{(\kappa)})^2\mathbf{h}_k^H(\bar{\boldsymbol{\Phi}}_k^{(\kappa)})^{-1}\bar{\boldsymbol{\Phi}}_k(\bar{\boldsymbol{\Phi}}_k^{(\kappa)})^{-1}\mathbf{h}_k,\quad
\end{IEEEeqnarray} 
where $\boldsymbol{\bar{\Phi}}_k^{(\kappa)}\triangleq\sum\nolimits_{\ell=1}^{k-1}(p_\ell^{(\kappa)})^2\bh_\ell\bh_\ell^H  + \sigma_{\mathtt{Ant}}^2\mathbf{I} + \alpha_1^{(\kappa)}\sigma^2_{\mathtt{R}}\mathbf{I}$. The proof is done by the fact that $h(p_k,\bar{\boldsymbol{\Phi}}_k)$ is a convex function in $(p_k,\bar{\boldsymbol{\Phi}}_k)$ \cite[Eq. (31)]{Dinh:JSAC:Dec2017}. Therefore, we can iteratively replace \eqref{MMReqi:c}   with the following convex constraint:
\begin{IEEEeqnarray}{rCl}\label{eq:approhfunction}
h^{(\kappa)}(p_k,\bar{\boldsymbol{\Phi}}_k) \geq 1/\psi_{1,k},\ \forall k\in\mathcal{K}.
\end{IEEEeqnarray}
We  rewrite \eqref{MMReqi:d} as $(\Re\{\bg_k\bw_k\})^2 \geq \frac{\sum_{\ell\in \K \backslash k}|\bg_{k}\bw_{\ell}|^2 + \sigma_k^2}{\psi_{2,k}}$, which can be convexified as
\begin{IEEEeqnarray}{lll}
	g^{(\kappa)}(\bw_k) \geq \frac{\sum_{\ell\in \K \backslash k}|\bg_{k}\bw_{\ell}|^2 + \sigma_k^2}{\psi_{2,k}},\ \forall k\in\mathcal{K},
	\label{std1}
	\end{IEEEeqnarray} 
upon the condition
\begin{IEEEeqnarray}{lll}
	\Re\{\bg_k^H\bw_k\} \geq 0,\ \forall k\in\mathcal{K},
	\label{std2}
	\end{IEEEeqnarray} 
where $g^{(\kappa)}(\bw_k)\triangleq 2\Re\{\bg_k\bwi_k\}\Re\{\bg_k\bw_k\}	- (\Re\{\bg_k\bwi_k\})^2$ is the concave approximation of $(\Re\{\bg_k\bw_k\})^2$ at $\bwi_k$.

We are now in a position to approximate \eqref{MMReqi:e} and \eqref{MMReqi:f}. It is true that 
\begin{subnumcases}{\label{MMReqi:e1} \eqref{MMReqi:e} \Leftrightarrow}
	a\sum_{k\in \K} \frac{p_k^2\|\bh_k\|_2^2}{\alpha_2} + \ln\bigl(\frac{1}{\vartheta}\bigr)\frac{1}{1-\beta} \geq \frac{ab}{1-\beta},\IEEEyessubnumber\label{MMReqi:e1a}\qquad\\
	 (\tau_2-1)\vartheta \geq (1+\vartheta)\theta,\IEEEyessubnumber\label{MMReqi:e1b}\\
	 \xi\theta -\xi\Omega(\tau_2-1) -  \tau_2 P_{\R}^{\mathtt{sta}} \geq P_{\R}^{\mathtt{BF}}(\bw),\quad\IEEEyessubnumber\label{MMReqi:e1c}
\end{subnumcases}
where  $\xi \triangleq \frac{\bar{P}_\mathtt{EH}^{\max}}{1-\Omega}$, and $\vartheta$ and $\theta$ are   slack variables. Constraint \eqref{MMReqi:e1c} is convex, while \eqref{MMReqi:e1a} and \eqref{MMReqi:e1b} still remain  non-convex. By applying the first-order Taylor series approximation to the non-convex parts of \eqref{MMReqi:e1a}  
and an approximation of bilinear function $(1+\vartheta)\theta$ of \eqref{MMReqi:e1b} \cite{Beck:JGO:10}, it follows that
\begin{IEEEeqnarray}{rCl}\label{MMReqi:e2} 
  a\mathcal{H}^{(\kappa)}(\bp,\alpha_2) + \tilde{f}^{(\kappa)}(\vartheta,\beta) &\geq&  \frac{ab}{1-\beta},\IEEEyessubnumber\label{MMReqi:e2a}\qquad\\
(\tau_2-1)\vartheta &\geq& \mathcal{B}^{(\kappa)}(1+\vartheta,\theta),\IEEEyessubnumber\label{MMReqi:e2b}
\end{IEEEeqnarray}
where $\mathcal{H}^{(\kappa)}(\bp,\alpha_2)\triangleq \sum_{k\in \K} \Big( \frac{2\pki\|\bh_k\|_2^2}{\alphadi}p_k - \frac{(\pki)^2\|\bh_k\|_2^2}{(\alphadi)^2}\alpha_2\Big)$,  $\tilde{f}^{(\kappa)}(\vartheta,\beta)\triangleq 2\ln\big(\mfrac{1}{\varthetai}\big)\mfrac{1}{1-\betai} - \mfrac{\vartheta}{\varthetai(1-\betai)} +\mfrac{1}{1-\betai} - \ln\big(\mfrac{1}{\varthetai}\big)\mfrac{1-\beta}{(1-\betai)^2}$ and $\mathcal{B}^{(\kappa)}(1+\vartheta,\theta)\triangleq 0.5\bigl(\frac{1+\vartheta^{(\kappa)}}{\theta^{(\kappa)}}\theta^2 + \frac{\theta^{(\kappa)}}{1+\vartheta^{(\kappa)}}(1+\vartheta)^2\bigr)$. Finally, we can transform \eqref{MMReqi:f} into  $\sum_{k\in \K} \frac{p_k^2\|\bh_k\|_2^2}{\alpha_2} \geq P^\mathtt{ACC}_{\min}\frac{\tau_1}{\beta}$, which is innerly approximated as
\begin{IEEEeqnarray}{rCl}\label{MMReqi:fconvex} 
  \mathcal{H}^{(\kappa)}(\bp,\alpha_2) \geq P^\mathtt{ACC}_{\min} \mathcal{B}^{(\kappa)}(\tau_1,1/\beta),
\end{IEEEeqnarray}
by following the same procedures as in \eqref{MMReqi:e2}. 

Summing up,  at iteration $\kappa+1$, we solve the following approximate convex program:
\begin{IEEEeqnarray}{rCl} \label{MMRapproconvex}
		 &&\max_{\mathbf{s},r}\quad  r  \IEEEyessubnumber\label{MMRapproconvex:a}\\
		&&\st\ \eqref{MMReqi:g}-\eqref{MMReqi:i},\eqref{MMReqi:bconvex},\eqref{eq:approhfunction},\eqref{std1},\eqref{std2},\eqref{MMReqi:e1c},\eqref{MMReqi:e2},\eqref{MMReqi:fconvex},\qquad\IEEEyessubnumber\label{MMRapproconvex:b}
	\end{IEEEeqnarray}
where $\mathbf{s}\triangleq\{\bp, \bw,\boldsymbol{\tau}, \boldsymbol{\alpha},\boldsymbol{\psi},\beta,\vartheta,\theta\}$ denotes the set of variables that needs to be updated in the next iteration. We can see that the main barrier in finding an initial feasible point to start the computational procedure for  \eqref{MMReqi} is due to constraint \eqref{MMReqi:f}. Therefore, we successively solve the following modified convex program of \eqref{MMRapproconvex}:
\begin{IEEEeqnarray}{rCl} \label{MMRapproconvexIN}
		 &&\max_{\mathbf{s},r}\quad  \eta \triangleq \mathcal{H}^{(\kappa)}(\bp,\alpha_2) - P^\mathtt{ACC}_{\min} \mathcal{B}^{(\kappa)}(\tau_1,1/\beta) \IEEEyessubnumber\label{MMRapproconvexIN:a}\\
		&&\st\ \eqref{MMReqi:g}-\eqref{MMReqi:i},\eqref{MMReqi:bconvex},\eqref{eq:approhfunction},\eqref{std1},\eqref{std2},\eqref{MMReqi:e1c},\eqref{MMReqi:e2},\qquad\IEEEyessubnumber\label{MMRapproconvexIN:b}
	\end{IEEEeqnarray}
until reaching $\eta \geq 0$. To efficiently solve \eqref{MMRapproconvexIN}, we first set $p_k^{(0)} = \sqrt{P_{\Sk}^{\max}},\forall k,\tau_1^{(0)}=\tau_2^{(0)}=\alpha_1^{(0)}=\alpha_2^{(0)}=2, \beta=0.5$, and randomly generate a sufficiently small value of $\bw^{(0)}$ to ensure that \eqref{MMReqi:e} is feasible. Other initial points can be found as $\psi_{1,k}^{(0)}=1/\gamma_{1,k}(\bp^{(0)},1/\alpha_1^{(0)}), \psi_{2,k}^{(0)}=1/\gamma_{2,k}(\bw^{(0)})$,  $\vartheta^{(0)}=1/\exp\bigl(-a(P_\mathtt{R}^{\mathtt{IN}}(\bp^{(0)},1/\alpha_2^{(0)},1-\beta^{(0)})-b)\bigr)$, and $\theta^{(0)} = (\tau_2^{(0)}-1)\vartheta^{(0)}/(1+\vartheta^{(0)})$ by setting inequalities \eqref{MMReqi:c}, \eqref{MMReqi:d}, \eqref{MMReqi:e1a} and \eqref{MMReqi:e1b} to equalities, respectively. The proposed iterative algorithm is summarized in \textbf{Algorithm \ref{alg_1}}.
\begin{algorithm}[t]
	\begin{algorithmic}[1]{ 
		\protect\caption{Proposed Iterative Algorithm for Solving  \eqref{MMR}}
		\label{alg_1}
		\global\long\def\algorithmicrequire{\textbf{Initialization:}}
		\REQUIRE  Set $\kappa:=0, \kappa':=0$ and randomly generate  $\mathbf{s}^{(0)}$.\\
		\global\long\def\algorithmicrequire{\textbf{Generating a feasible point for \eqref{MMReqi}}: }
		\REQUIRE
		\REPEAT
		\STATE Solve  \eqref{MMRapproconvexIN} to obtain the optimal solution $\mathbf{s}^\star$.
		\STATE Update $\mathbf{s}^{(\kappa'+1)}:=\mathbf{s}^\star$ and set $\kappa':=\kappa'+1.$
		\UNTIL $\eta \geq 0$\\
		\STATE Set $\mathbf{s}^{(0)}:=\mathbf{s}^{(\kappa')}$.\\
		\global\long\def\algorithmicrequire{\textbf{Solving \eqref{MMReqi}}: }
		\REQUIRE
    \REPEAT
		\STATE Solve \eqref{MMRapproconvex} to obtain the optimal solution $\mathbf{s}^\star.$
		\STATE Update $\mathbf{s}^{(\kappa+1)}:=\mathbf{s}^{\star}$and set $\kappa:=\kappa+1.$
		\UNTIL Convergence \\ 
		\STATE \textbf{Output:} $(\mathbf{p},\bw,\tau,\alpha,\beta):=\bigl(\mathbf{p}^{(\kappa)},\bw^{(\kappa)},\frac{1}{\tau_1^{(\kappa)}},\frac{1}{\alpha_1^{(\kappa)}},\beta^{(\kappa)}\bigl)$.}
\end{algorithmic} \end{algorithm}

\begin{remark}
We note that problem \eqref{MMRapproconvex} can be transformed into an SOCP, where modern convex solvers are very efficient. The key is to further transform constraints \eqref{MMReqi:g}, \eqref{MMReqi:h}, \eqref{eq:approhfunction} and \eqref{MMReqi:fconvex} into SOC ones:
\begin{subnumcases}{}
\eqref{MMReqi:g} \Leftrightarrow \tilde{\tau}_1 + \tilde{\tau}_2 \leq 1,\ \&\ \tilde{\tau}_1\tau_1 \geq 1,\ \&\  \tilde{\tau}_2\tau_2 \geq 1,\nonumber\\
\eqref{MMReqi:h} \Leftrightarrow \tilde{\alpha}_1 + \tilde{\alpha}_2 \leq 1,\ \&\ \tilde{\alpha}_1\alpha_1 \geq 1,\ \&\  \tilde{\alpha}_2\alpha_2 \geq 1,\nonumber\\
	\eqref{eq:approhfunction} \Leftrightarrow h^{(\kappa)}(p_k,\bar{\boldsymbol{\Phi}}_k) \geq \tilde{\psi}_{1,k},\ \&\ \psi_{1,k}\tilde{\psi}_{1,k} \geq 1, \forall k\in\mathcal{K},\qquad\nonumber\\
	\eqref{MMReqi:fconvex}\Leftrightarrow\mathcal{H}^{(\kappa)}(\bp,\alpha_2) \geq P^\mathtt{ACC}_{\min} \mathcal{B}^{(\kappa)}(\tau_1,\tilde{\beta}),\ \&\ \beta\tilde{\beta} \geq 1, \nonumber
\end{subnumcases}
where $\tilde{\tau}_i,\tilde{\alpha}_i, \forall i\in\mathcal{I}, \tilde{\psi}_{1,k},\forall k\in\mathcal{K}$ and $\tilde{\beta}$ are slack variables.
\end{remark}

\textit{ Convergence and Complexity Analysis}: We can see that all the convex approximations in \eqref{MMRapproconvex} satisfy the IA properties listed in \cite{Beck:JGO:10}. In other words, the optimal solution obtained at iteration $\kappa$ of \textbf{Algorithm \ref{alg_1}} is also feasible for problem \eqref{MMRapproconvex} at iteration $\kappa + 1$.  It implies that \textbf{Algorithm \ref{alg_1}} produces a sequence $\mathbf{s}^{(\kappa)}$ of improved
points of \eqref{MMR}, which converges to at least  a local optimum.
Problem \eqref{MMRapproconvex} involves $(6K + 7)$ conic constraints and $(KN + 3K +8)$ scalar decision variables. Thus,  the worst-case computational complexity in each iteration of \textbf{Algorithm \ref{alg_1}} is $\mathcal{O}\bigl((6K)^{0.5}(KN + 3K)^3\bigr)$. Similarly, the complexity of  \eqref{MMRapproconvexIN} for finding an initial feasible point is  $\mathcal{O}\bigl((6K)^{0.5}(KN + 3K)^3\bigr)$.

{\hili \begin{remark}
	The channel state information (CSI) between the relay and sources/destinations, as well as the controlling signal, can be exchanged via dedicated channels, and thus the algorithm is simply executed at the relay.  Moreover, Algorithm \ref{alg_1} can be slightly modified to solve the worst-case robust optimization problem, where the bounded CSI error is taken into account.
\end{remark}}

\section{Numerical Results}
We consider the relay network as shown  in Fig. \ref{fig:systemodel}, in which the distances from  each source to relay and from the relay to each destination are set to be  10 m and 15 m, respectively. The networks parameters are set as $K=4$, $\sigma_{\mathtt{Ant}}^2=\sigma_k^2=-70$ dBm, $\sigma^2_{\mathtt{R}}=-50$ dBm, $P_\R^\mathtt{sta}=1$ $\mu$W, {$P^\mathtt{ACC}_{\min}= 0.27$ $\mu$W and $P^\mathtt{DCC}_{\min}=47.64$ $\mu$W} \cite{SalmanVLSI18,ShiTWC2014}.  The parameters of nonlinear EH model are $\bar{P}^{\max}_\mathtt{EH}=0.2$ mW, $a$ = 6400 and $b = 0.003$ \cite{XiongTWC17}. All channels are assumed to undergo Rayleigh fading with the path-loss exponent of 3.5.  All source nodes are assumed to have the same power budget, i.e., $P_{\mathtt{S}}^{\max}=P_{\Sk}^{\max},\forall k$. We use the YALMIP toolbox with the SeDuMi solver to solve the convex problems.
For benchmarking purpose, we compare the performance of Algorithm \ref{alg_1} with the use of DCC in \eqref{MMRDCC} and four other suboptimal schemes:
($i$) 	 ``Equal Block Time (EBT)" with $\tau = 0.5$;
($ii$) 	 ``Equal Power Splitting (EPS)" with $\alpha = 0.5$;
($iii$)  ``Equal Block Time and Equal Power Splitting (EBT-EPS)" with $\tau=\alpha=0.5$; and
	($iv$)  ``Non-SIC" without using the SIC at the ID receiver.

\begin{figure}[!t]
	\centering
		\begin{subfigure}[$N = 4$.]{
			\includegraphics[width=0.21\textwidth]{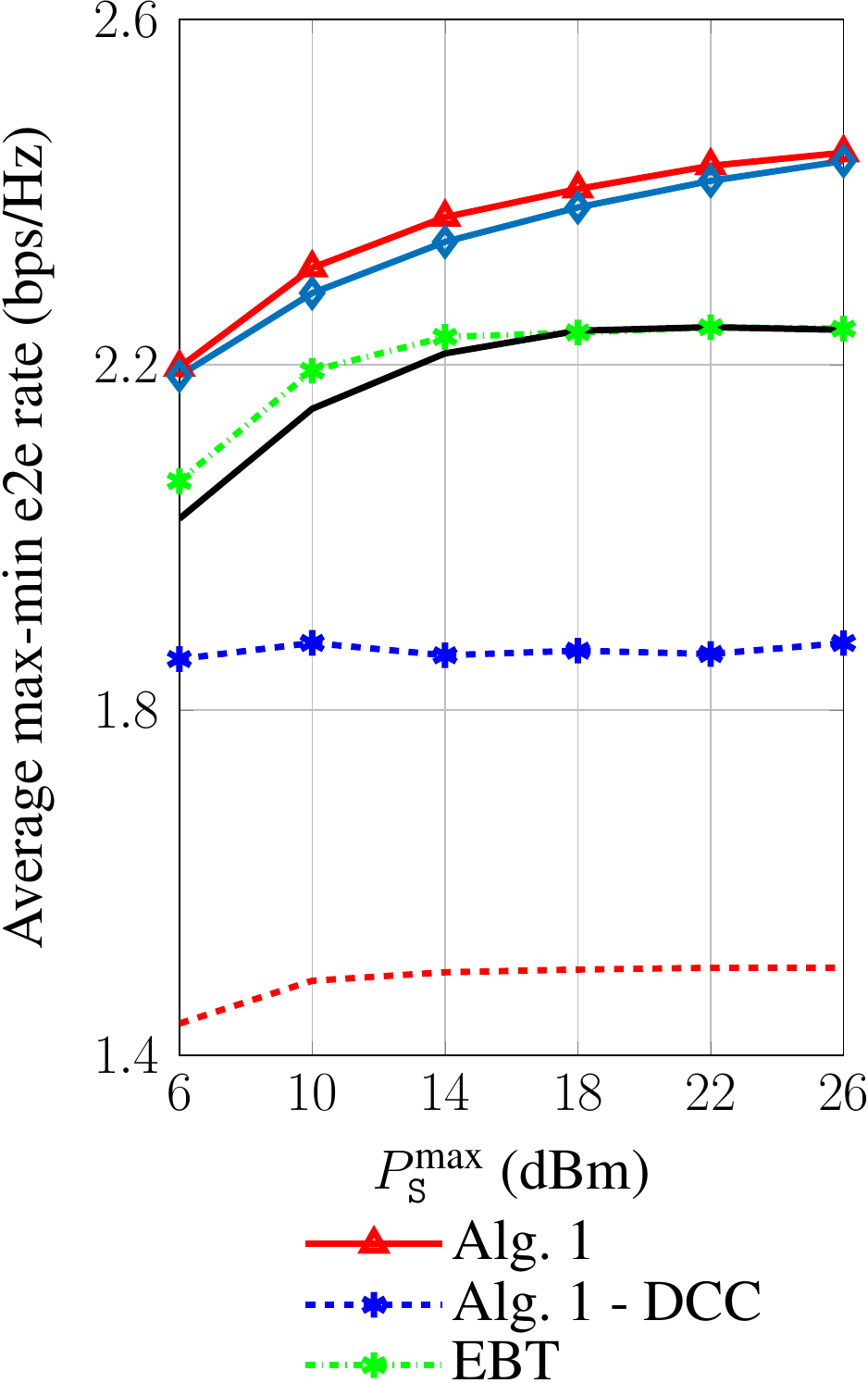}
			\label{fig: power}}
		\end{subfigure}
	\hfill
		\begin{subfigure}[$P^\text{max}_\mathtt{S} = 18$ dBm.]{
			\includegraphics[width=0.20\textwidth]{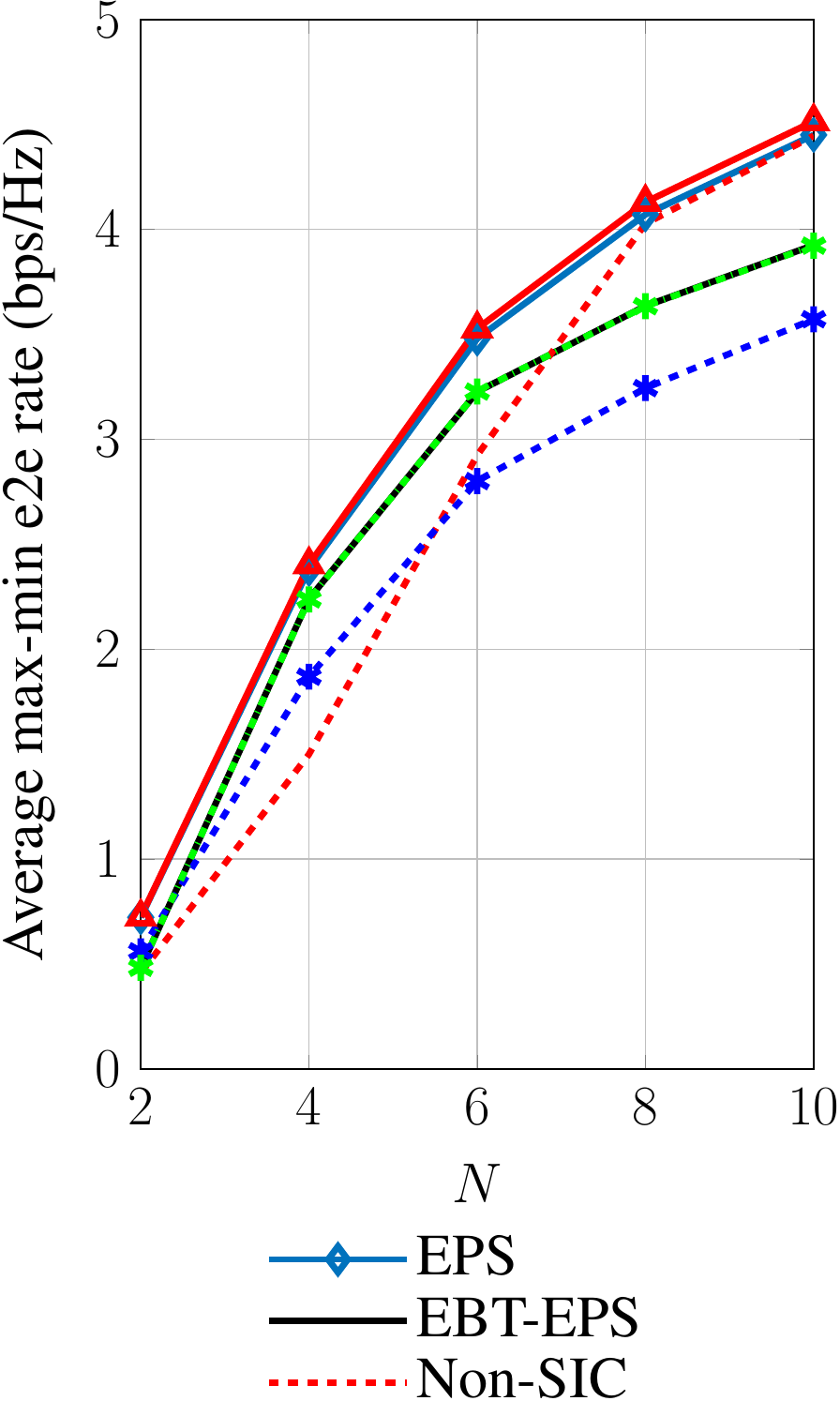}
			\label{fig: antenna}}
		\end{subfigure}
		\caption{Average max-min e2e rate versus  $P_{\mathtt{S}}^{\max}$ and $N$.}	
		\label{fig: power and antenna}
\end{figure}

We plot the average max-min e2e rate versus  $P_{\mathtt{S}}^{\max}$ and $N$ in Figs. \ref{fig: power} and \ref{fig: antenna}, respectively. As can be observed,  the proposed Algorithm \ref{alg_1} indeed shows better performance compared to the others in all cases. The results also confirm that significant performance gain can be achieved by jointly optimizing involved parameters, compared to the EBT, EPS and EBT-EPS schemes. In addition, the performance
gaps between Algorithm \ref{alg_1} and EPS, and between EBT and EBT-EPS  are not significant for high values of $P_{\mathtt{S}}^{\max}$ and $N$,  implying that $\alpha=0.5$ is a near-optimal solution. Moreover, the use of ACC shows its effectiveness since Algorithm \ref{alg_1} can achieve superior performance compared to the use of DCC. This is attributed to the fact that there is no EH conversion loss and $P^\mathtt{DCC}_{\min} \gg P^\mathtt{ACC}_{\min}$. In Fig. \ref{fig: power}, the non-SIC scheme provides the worst performance due to  severe interference at the relay, thus reaching a saturated value quickly when $P_{\mathtt{S}}^{\max} \geq 10$ dBm. However, the performance of the non-SIC approach catches
up with that of Algorithm \ref{alg_1} in Fig. \ref{fig: antenna}, as $N$ increases. The reason is that a relay with more antennas is able to combat the interference more effectively. These observations further validate the benefits of the proposed  PSR architecture-enabled SIC and ACC at the relay.

\begin{figure}[!t] 
	\centering
	\begin{subfigure}[Convergence of Alg. \ref{alg_1} with one random channel realization.]
		{
			\includegraphics[width=.27\textwidth, trim={0cm, 0cm, -0.0cm, 0cm}]{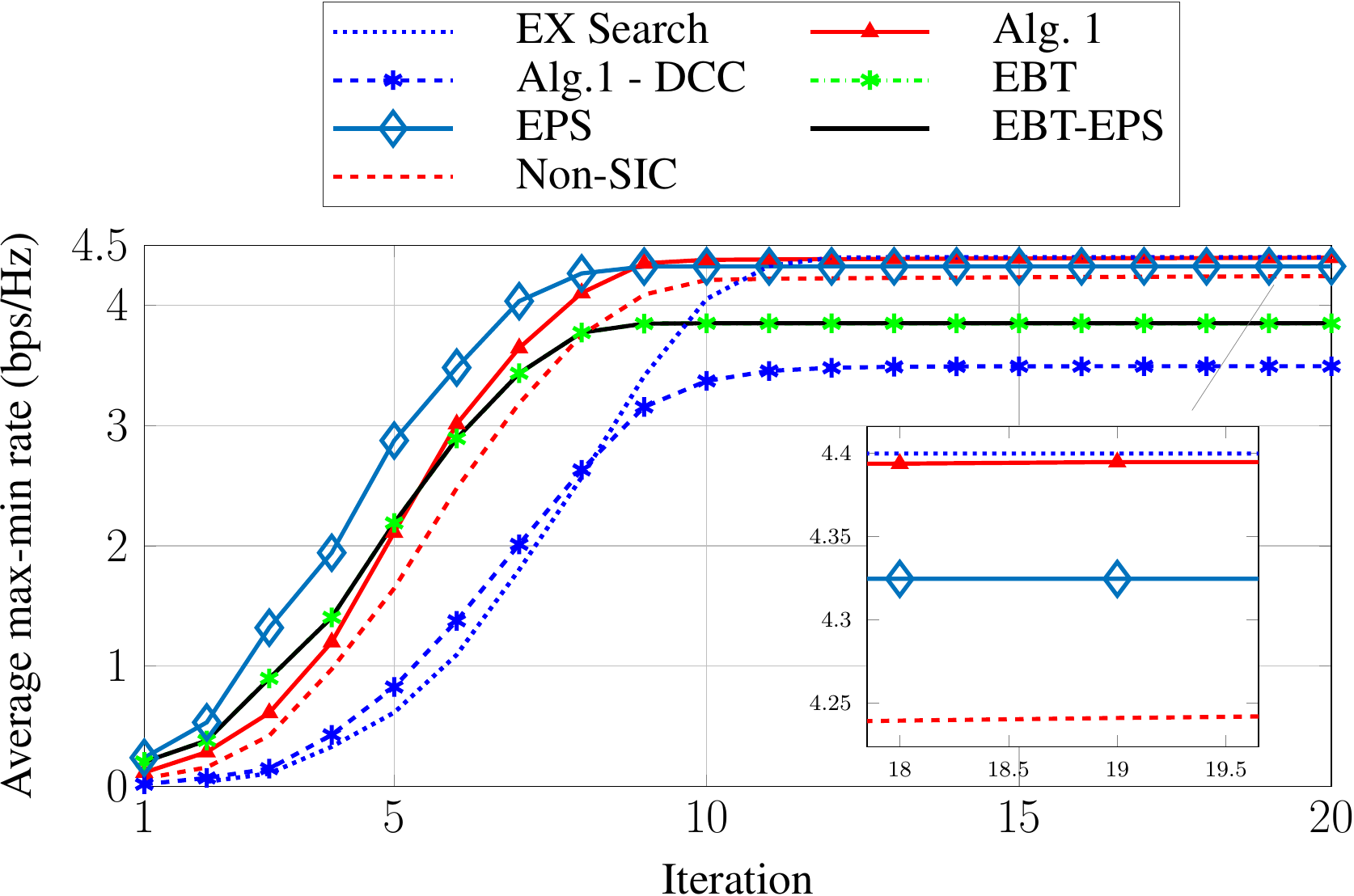}
			\vspace{-0pt}
			\label{fig: convergence}
		}
	\end{subfigure}
	\hfill
	\begin{subfigure}[Cumulative distribution function (CDF) of the  max-min e2e rate.]
		{\includegraphics[width=.28\textwidth, trim={0cm, 0cm, 0cm, 0cm}]{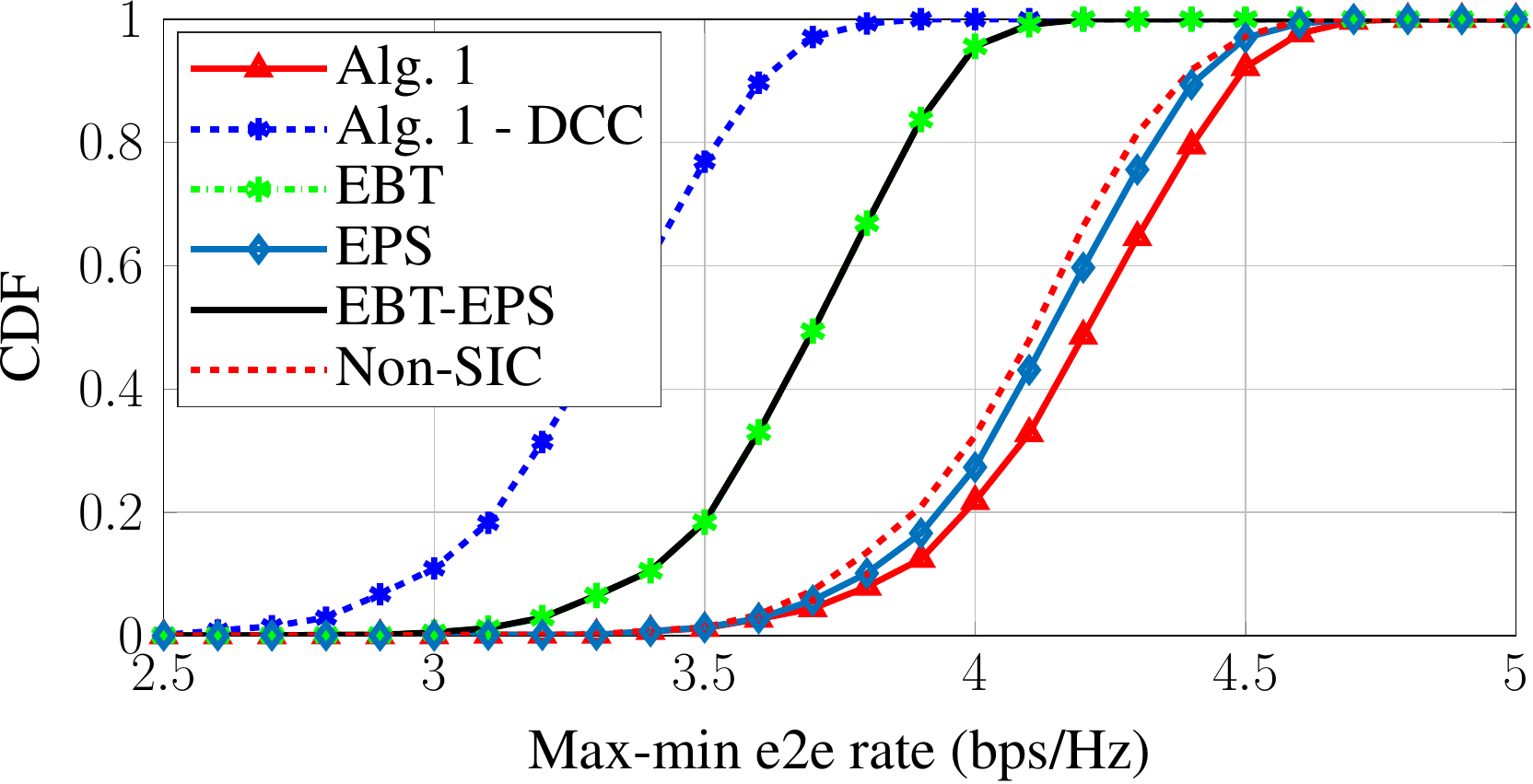}
			\vspace{-0pt}
			\label{fig: CDF}
		}
	\end{subfigure}
	\vspace{-5pt}
	\caption{Performance comparison for  different resource allocation schemes ($P^{\max}_\mathtt{S} = 18 \text{ dBm}$ and $ N=8$).}\label{fig: cdf and convergence}
\end{figure}

 Fig. \ref{fig: convergence} depicts  the convergence behavior of the proposed algorithm with different resource allocation schemes over a random channel. We have numerically observed that the proposed algorithm requires a maximum of two iterations to output an initial feasible point. As can be seen, Algorithm \ref{alg_1} converges after a few iterations and achieves the max-min rates very close to  the exhaustive search (i.e., EX search) method.   Another observation is that the EBT, EPS and EBT-EPS schemes converge faster due to less optimization variables, but their performance is inferior to that of Algorithm \ref{alg_1}. As expected in Fig. \ref{fig: CDF},   the proposed Algorithm \ref{alg_1} is able to maintain better e2e rate fairness among all user pairs, compared to other schemes.

\section{Conclusion}
In this letter, we  proposed a new and practical  PSR architecture for multi-pair wireless-powered DF relaying networks, which enables  SIC at the relay and allows to directly use the harvested AC power for activating computational blocks. We first
formulated the problem of max-min e2e rate fairness among all user
pairs by jointly designing the power control, beamforming, fractional time  and power splitting ratios, and then developed a low-complexity
solution by employing  the IA optimization framework. The effectiveness of the proposed method was demonstrated by  numerical results.

\appendices
\renewcommand{\thesectiondis}[2]{\Alph{section}:}
\section*{Appendix: Proof of Lemma  \ref{lemma:1}} \label{app: lemma:1}
\renewcommand{\theequation}{\ref{app: lemma:1}.\arabic{equation}}\setcounter{equation}{0}
We first note that constraints \eqref{MMReqi:g} and \eqref{MMReqi:h} must hold with equalities at  optimum. We now prove Lemma  \ref{lemma:1} by verifying that constraints \eqref{MMReqi:b}-\eqref{MMReqi:d} are active at optimum by contradiction. Let $(\bp^\star, \bw^\star,\boldsymbol{\tau}^\star, \boldsymbol{\alpha}^\star,\boldsymbol{\psi}^\star,\beta^\star,r^\star)$ be an optimal solution of \eqref{MMReqi}. Suppose that \eqref{MMReqi:b}-\eqref{MMReqi:d} are inactive, i.e., $\ln\bigl(1+ 1/\psi_{i,k}^\star\bigr)/\tau_i^\star > r^\star$, $\gamma_{1,k}(\bp^\star,1/\alpha_1^\star) > 1/\psi_{1,k}^\star$ and $\gamma_{2,k}(\bw^\star)  >  1/\psi_{2,k}^\star$ for some $i,k$. There exists $\psi_{i,k}'$ such that $\psi_{i,k}' < \psi_{i,k}^\star$, $\gamma_{1,k}(\bp^\star,1/\alpha_1^\star) > 1/\psi_{1,k}'$ and $\gamma_{2,k}(\bw^\star)  >  1/\psi_{2,k}'$. Then, there may also exist a positive constant $\Delta r >0$ to satisfy $\ln\bigl(1+ 1/\psi_{i,k}'\bigr)/\tau_i^\star = r^\star + \Delta r$. 
As a result, $r^\star + \Delta r$ and $\psi_{i,k}'$ are also feasible to \eqref{MMReqi}, yielding a strictly larger
objective. This contradicts the optimality assumption of $(\bp^\star, \bw^\star,\boldsymbol{\tau}^\star, \boldsymbol{\alpha}^\star,\boldsymbol{\psi}^\star,\beta^\star,r^\star)$, and thus completes the proof.

\bibliographystyle{IEEEtran}
 \bibliography{Journal}

\begin{thebibliography}{10}
\providecommand{\url}[1]{#1}
\csname url@samestyle\endcsname
\providecommand{\newblock}{\relax}
\providecommand{\bibinfo}[2]{#2}
\providecommand{\BIBentrySTDinterwordspacing}{\spaceskip=0pt\relax}
\providecommand{\BIBentryALTinterwordstretchfactor}{4}
\providecommand{\BIBentryALTinterwordspacing}{\spaceskip=\fontdimen2\font plus
\BIBentryALTinterwordstretchfactor\fontdimen3\font minus
  \fontdimen4\font\relax}
\providecommand{\BIBforeignlanguage}[2]{{%
\expandafter\ifx\csname l@#1\endcsname\relax
\typeout{** WARNING: IEEEtran.bst: No hyphenation pattern has been}%
\typeout{** loaded for the language `#1'. Using the pattern for}%
\typeout{** the default language instead.}%
\else
\language=\csname l@#1\endcsname
\fi
#2}}
\providecommand{\BIBdecl}{\relax}
\BIBdecl

\bibitem{Rankov:JSAC:07}
B.~{Rankov} and A.~{Wittneben}, ``Spectral efficient protocols for half-duplex
  fading relay channels,'' \emph{IEEE J. Select. Areas in Commun.}, vol.~25,
  no.~2, pp. 379--389, Feb. 2007.

\bibitem{NasirTWC13}
A.~A. {Nasir}, X.~{Zhou}, S.~{Durrani}, and R.~A. {Kennedy}, ``Relaying
  protocols for wireless energy harvesting and information processing,''
  \emph{IEEE Trans. Wireless Commun.}, vol.~12, no.~7, pp. 3622--3636, July
  2013.

\bibitem{ClerckxJSAC19}
B.~{Clerckx} \emph{et~al.}, ``Fundamentals of wireless information and power
  transfer: From {RF} energy harvester models to signal and system designs,''
  \emph{IEEE J. Select. Areas Commun.}, vol.~37, no.~1, pp. 4--33, Jan. 2019.

\bibitem{COMML:Octavia}
R.~{Wang} \emph{et~al.}, ``Optimal power allocation for full-duplex underwater
  relay networks with energy harvesting: A reinforcement learning approach,''
  \emph{IEEE Wireless Commun. Lett.}, pp. 1--1, 2019.

\bibitem{BenkhelifaJSAC16}
F.~{Benkhelifa} \emph{et~al.}, ``Sum-rate enhancement in multiuser {MIMO}
  decode-and-forward relay broadcasting channel with energy harvesting
  relays,'' \emph{IEEE J. Select. Areas Commun.}, vol.~34, no.~12, pp.
  3675--3684, Dec. 2016.

\bibitem{ZhangTWC17}
L.~{Zhang}, Y.~{Cai}, M.~{Zhao}, B.~{Champagne}, and L.~{Hanzo}, ``Nonlinear
  {MIMO} transceivers improve wireless-powered and self-interference-aided
  relaying,'' \emph{IEEE Trans. Wireless Commun.}, vol.~16, no.~10, pp.
  6953--6966, Oct. 2017.

\bibitem{ChuTCOM18}
M.~{Chu} \emph{et~al.}, ``On the design of power splitting relays with
  interference alignment,'' \emph{IEEE Trans. Commun.}, vol.~66, no.~4, pp.
  1411--1424, Apr. 2018.

\bibitem{WanESL17}
T.~{Wan}, Y.~{Karimi}, M.~{Stanacevic}, and E.~{Salman}, ``Perspective
  paper-{C}an {AC} computing be an alternative for wirelessly powered {IoT}
  devices?'' \emph{IEEE Embed. Syst. Lett.}, vol.~9, no.~1, pp. 13--16, Mar.
  2017.

\bibitem{SalmanVLSI18}
E.~Salman, M.~Stanacevic, S.~Das, and P.~M. Djuric, ``Leveraging {RF} power for
  intelligent tag networks,'' in \emph{Proc. ACM Great Lakes Symposium on
  VLSI}, Chicago, IL, USA, May 2018, pp. 329--334.

\bibitem{TranCOMML19}
H.-V. {Tran} and G.~{Kaddoum}, ``Robust design of {AC} computing-enabled
  receiver architecture for {SWIPT} networks,'' \emph{IEEE Wireless Commun.
  Lett.}, vol.~8, no.~3, pp. 801--804, June 2019.

\bibitem{NguyenACCESS2019}
V.-D. {Nguyen} \emph{et~al.}, ``An efficient design for {NOMA}-assisted
  {MISO-SWIPT} systems with {AC} computing,'' \emph{IEEE Access}, vol.~7, pp.
  97\,094--97\,105, 2019.

\bibitem{TsePramodBook05}
D.~Tse and P.~Viswanath, \emph{Fundamentals of Wireless Communication}.\hskip
  1em plus 0.5em minus 0.4em\relax New York, NY, USA: Cambridge University
  Press, 2005.

\bibitem{XiongTWC17}
K.~{Xiong}, B.~{Wang}, and K.~J.~R. {Liu}, ``Rate-energy region of {SWIPT} for
  {MIMO} broadcasting under nonlinear energy harvesting model,'' \emph{IEEE
  Trans. Wireless Commun.}, vol.~16, no.~8, pp. 5147--5161, Aug. 2017.

\bibitem{Beck:JGO:10}
A.~Beck, A.~Ben-Tal, and L.~Tetruashvili, ``A sequential parametric convex
  approximation method with applications to nonconvex truss topology design
  problems,'' \emph{J. Global Optim.}, vol.~47, no.~1, pp. 29--51, May 2010.

\bibitem{ShiTWC2014}
Q.~{Shi}, L.~{Liu}, W.~{Xu}, and R.~{Zhang}, ``Joint transmit beamforming and
  receive power splitting for {MISO SWIPT} systems,'' \emph{IEEE Trans.
  Wireless Commun.}, vol.~13, no.~6, pp. 3269--3280, June 2014.

\bibitem{NguyenJSAC18}
V.-D. Nguyen \emph{et~al.}, ``A new design paradigm for secure full-duplex
  multiuser systems,'' \emph{IEEE J. Select. Areas Commun.}, vol.~36, no.~7,
  pp. 1480--1498, July 2018.

\bibitem{Dinh:JSAC:Dec2017}
V.-D. Nguyen, H.~D. Tuan, T.~Q. Duong, H.~V. Poor, and O.-S. Shin, ``Precoder
  design for signal superposition in {MIMO-NOMA} multicell networks,''
  \emph{IEEE J. Select. Areas Commun.}, vol.~35, no.~12, pp. 2681--2695, Dec.
  2017.

\end{thebibliography}
\end{document}